\documentclass[conference]{IEEEtran}
\usepackage{graphicx}
\usepackage{cite}
\usepackage{enumerate}
\usepackage{array}
\usepackage{amsmath}
\usepackage{amssymb}
\usepackage{subeqnarray}
\usepackage{cases}
\usepackage{bm}
\usepackage{multirow}
\usepackage{epsfig}
\usepackage{epstopdf}
\usepackage{xcolor}
\usepackage{subfigure}

\newtheorem{theorem}{Theorem}
\newtheorem{lemma}{Lemma}
\newtheorem{proposition}{Proposition}

\newenvironment{proof}{{\bf Proof:}}{\hfill\rule{2mm}{2mm}}

\newcommand{\E}[1]{\mathrm{E}\left[#1\right]}

\begin{document}
\title{Beyond Age: Urgency of Information for Timeliness Guarantee in Status Update Systems}
\author{\IEEEauthorblockN{Xi~Zheng, Sheng~Zhou, Zhisheng~Niu}

\IEEEauthorblockA{Beijing National Research Center for Information Science and Technology\\
Department of Electronic Engineering, Tsinghua University, Beijing 100084, P.R. China\\
zhengx14@mails.tsinghua.edu.cn,
\{sheng.zhou, niuzhs\}@tsinghua.edu.cn}}

\maketitle

\begin{abstract}
Timely status updating is crucial for future applications that involve remote monitoring and control, such as autonomous driving and Industrial Internet of Things (IIoT). Age of Information (AoI) has been proposed to measure the freshness of status updates. However, it is incapable of capturing critical systematic context information that indicates the time-varying importance of status information, and the dynamic evolution of status. In this paper, we propose a context-based metric, namely the Urgency of Information (UoI), to evaluate the timeliness of status updates. Compared to AoI, the new metric incorporates both time-varying context information and dynamic status evolution, which enables the analysis on context-based adaptive status update schemes, as well as more effective remote monitoring and control. The minimization of average UoI for a status update terminal with an updating frequency constraint is investigated, and an update-index-based adaptive scheme is proposed. Simulation results show that the proposed scheme achieves a near-optimal performance with a low computational complexity. 
\end{abstract}

\section{Introduction}
Emerging applications in next-generation networks require timely and reliable status updates to reduce information asymmetry and enable collaboration among various ends \cite{network}\cite{tactile}. Taking vehicular networks as an example, vehicles need to timely exchange position, velocity, acceleration and driving intention information to enable driving assistance applications, such as collision avoidance, intersection scheduling, and even autonomous driving. 
The timeliness requirement for status updates varies according to the non-uniform status evolution. Generally speaking, when the status changes rapidly, more frequent information delivery is required. Timeliness requirement is also related to context information \cite{context}\cite{context-survey}, which includes all the knowledge about the underlying system that determines how important the status is. For example, when the system is at a critical situation, its status should be more frequently updated. Otherwise, insufficient status updates will hinder the effectiveness of the applications, resulting in unacceptable performance degradation. On the other hand, excessive status information deliver will bring little marginal performance gain while consuming extra wireless and energy resources. Therefore, to ensure the timeliness of information delivery and the effectiveness of status-based control, status updates should adapt to the context of the system and the non-uniform evolution of status. 

To cope with varying context and the non-uniform status evolution, this paper proposes a new metric named the \emph{Urgency of Information} (UoI), defined as the product of context-aware weight and the systematic cost resulted from the inaccuracy of status estimation. The context-aware weight represents the importance of status information, while the cost identifies how inaccurate existing information is. Therefore, UoI is able to characterize the timeliness of status updates by differentiating the importance of status according to context and considering the non-uniform changes in status. 

There have been several metrics proposed to evaluate the timeliness of status updates. Age of information (AoI) \cite{aoi} is defined as the time elapsed since the generation of the most up-to-date packet received. It focuses on the freshness of information, which is linear with time, and is irrelevant to context and status evolution. AoI has been studied extensively in recent literatures \cite{aoi}--\!\!\cite{NegExpo}. To overcome the linearity limitation of AoI, ref. \cite{8000687} characterizes the non-linear loss caused by information staleness as a function of AoI. However, by investigating the mean square error (MSE) minimization in the remote estimation \cite{wiener}--\!\!\cite{EffectiveAge}, it is proved that AoI-based status update scheme is not MSE-optimal. Authors of \cite{AoS} propose Age of Synchronization (AoS) to measure the length of time when the latest status information is not synchronized with the actual status, which is useful in applications like caching design. A similar concept named Age of Incorrect Information (AoII) is also introduced in \cite{AoII}. However, none of these metrics takes context into consideration. 

In this paper, we investigate how to exploit context information and real-time status information to design adaptive status update schemes. The contributions of this paper are summarized as follows:
\begin{enumerate}
\item A new metric, namely the Urgency of Information, is proposed, through which both context-based importance and the non-uniform evolution of status are incorporated to represent the timeliness of status updates. 
\item An adaptive updating scheme to reduce the UoI of a status update terminal with an average updating frequency constraint is proposed. The adaptive scheme decides on whether to sends a status update by comparing a history-related virtual queue length and a present-based update index, and is proved to obtain a bounded UoI. It is shown that the adaptive scheme is able to achieve a near-optimal performance in the simulations. 
\end{enumerate}

The remainder of this paper is organized as follows. Section II introduces the concept of UoI. The problem of status updates with a constrained average updating frequency is formulated in Section III. The design of an adaptive updating scheme is investigated in Section IV. In Section V, the performance of various policies is illustrated with simulation results. Section VI concludes the paper. 
                                                                                                                                                                                                                 
\section{Urgency of Information: A New Metric}
The timeliness of status updates is defined by the performance of the underlying remote monitoring and control applications. We first introduce the concept of urgency of information, then investigate the control performance of a fundamental remote control system to get a preliminary understanding of how the control performance is related to the staleness of status information. 

\subsection{Definition of UoI}
Denoting the difference between the actual status $x(t)$ and the estimated status $\hat{x}(t)$ at time $t$ by $Q(t) = x(t) - \hat{x}(t)$, the system cost caused by the inaccuracy of status estimation is defined by $\delta(Q(t))$, where $\delta(\cdot)$ is a non-negative even function (e.g., norms, quadratic function). The context-aware importance of status is evaluated by weight $\omega(t)$: if the context in the system presents a high requirement on estimation accuracy, the corresponding weight $\omega(t)$ will be large, and vice versa. 

To characterize the context-aware timeliness for status updates in remote control systems, we propose a new metric named the \emph{Urgency of Information} (UoI), which is previously named context-aware information lapse in \cite{previous}. The UoI is defined as the product of cost $\delta(Q(t))$ and context-aware weight $\omega(t)$. Mathematically, the UoI is expressed as
\begin{equation}
\label{eqn:def}
F(t) = \omega(t)\delta(Q(t)). 
\end{equation}
Denoting the temporal derivative $\frac{\mathrm{d}}{\mathrm{d}t}Q(t)$ of error by $A(t)$, error $Q(t)$ is equivalently written as 
\begin{eqnarray}
\label{eqn:q1}
Q(t) = \int_{g(t)}^tA_\tau\mathrm{d}\tau,
\end{eqnarray}
where $g(t)$ is the generation time stamp of the most up-to-date status update packet that is received before $t$. Note that if the cost $\delta(Q(t))$ increases in a constant rate (i.e., $A(t) = 1$) over time and the weight $\omega(t)$ is time-invariant, the UoI is equivalent to the conventional AoI. 
Moreover, there are other cases where UoI is equivalent to existing metrics: 
\begin{enumerate}
\item If the weight $\omega(t)$ is time-invariant, $A(t) = 1$ and mapping $\delta(\cdot)$ is defined as an AoI-penalty functions over $\mathbb{R}_+$, the UoI equals to the non-linear AoI defined in \cite{8000687}. 
\item If the weight $\omega(t)$ is time-invariant and $\delta(x) = x^2$, the UoI equals to the squared error. 
\end{enumerate}

The discrete-time version of the UoI is accordingly formulated as $F_t = \omega_t\delta\left(\sum_{\tau=g_t}^{t-1}A_{\tau}\right)$, where $A_t$ is the increment of error in time slot $t$. The recursive relationship is expressed as 
\begin{equation}
\label{eqn:dis}
Q_{t+1} = \left(1-D_t\right)Q_t + A_t + D_t\sum_{\tau=g_{t+1}}^{t-1}A_{\tau}, 
\end{equation}
where $D_t=1$ indicates that there is a successful status delivery in the $t$-th slot; otherwise $D_t = 0$. 

When the time-scale of status updates is much larger than the packet transmitting time, i.e., status information can be instantaneously obtained and delivered by the terminal whenever it is scheduled. In this case, Eq. (\ref{eqn:dis}) is written as $Q_{t+1} = \left(1-D_t\right)Q_t + A_t$.
\subsection{Tracking Control of Linear Systems}

Consider a remote control system where a controller remotely sends control actions to a terminal based on previous status feedbacks. The status of the terminal at the $t$-th time slot is denoted as $x_t$, and the dynamic function of status evolution is $x_t = ax_{t-1} + bv_t + r_t$, where $v_t$ is the control variable at the $t$-th time slot, and noise $r_t$ is an independent and identically distributed (i.i.d.) random variable with finite variance. Without loss of generality, it is assumed that the expected value of $r_t$ is zero. The controller decides on $v_t$ in order to keep terminal status $\{x_t|t\in\mathcal{N}\}$ as close to the desired status $\{y_t|t\in\mathcal{N}\}$ as possible. The control performance is evaluated by the weighted squared tracking error $\omega_t(x_t-y_t)^2$, where the context-based weight $\omega_t$ implies the importance of control performance at the $t$-th time slot. The objective is to minimize the weighted squared tracking error:
\begin{eqnarray}
\label{program:LQC}
&\min_{v_t}	& \limsup_{T\to\infty}\frac{1}{T}\sum_{t=0}^{T-1}\E{\omega_{t}\left(x_t - y_t\right)^2}. 
\end{eqnarray}
Due to information staleness, the controller might not be able to obtain current status $x_t$. At each time slot, the controller first estimates the status of the terminal based on historical status updates and control actions, and makes the optimal control decision $v^*_t$ based on the estimation $\hat{x}_{t-1}$. 

\begin{proposition}
The original problem (\ref{program:LQC}) is equivalent to minimizing the weighted status estimation error $$\limsup_{T\to\infty}\frac{1}{T}\sum_{t=0}^{T-1}\E{\omega_{t}\left(x_{t-1} - \hat{x}_{t-1}\right)^2}.$$
\end{proposition}
\begin{proof}
See Appendix A. 
\end{proof}

According to Proposition 1, the tracking problem of linear systems is equivalent to minimizing weighted estimation error in remote monitoring. The estimation error comes from the unpredictable nature of terminal status. As long as the wireless channel is not perfect or channel resource is limited, the controller can not get full knowledge of the actual status information each time it makes control actions, which harms the effectiveness of control. To reduce the estimation error, a sophisticated status updating scheme is in need to deliver status information timely by adapting to the context information and status evolution. 

\section{System Model}
Consider a wireless communication system with a terminal that constantly updates its status to a fusion center, as illustrated in Fig. \ref{fig:single}. Due to limited channel resource and energy supply, the frequency of sending status updates to the fusion can not exceed $\rho$. The decision of status updates at the $t$-th time slot is denoted by $U_t$, where $U_t=1$ means that the terminal sends its status at the $t$-th slot; otherwise $U_t=0$. Status packets are transmitted through a block fading channel with success probability $p$. The state of channel being good at the $t$-th slot is represented by $S_t=1$; otherwise $S_t=0$. 

\begin{figure}
\centering
\includegraphics[width=3in]{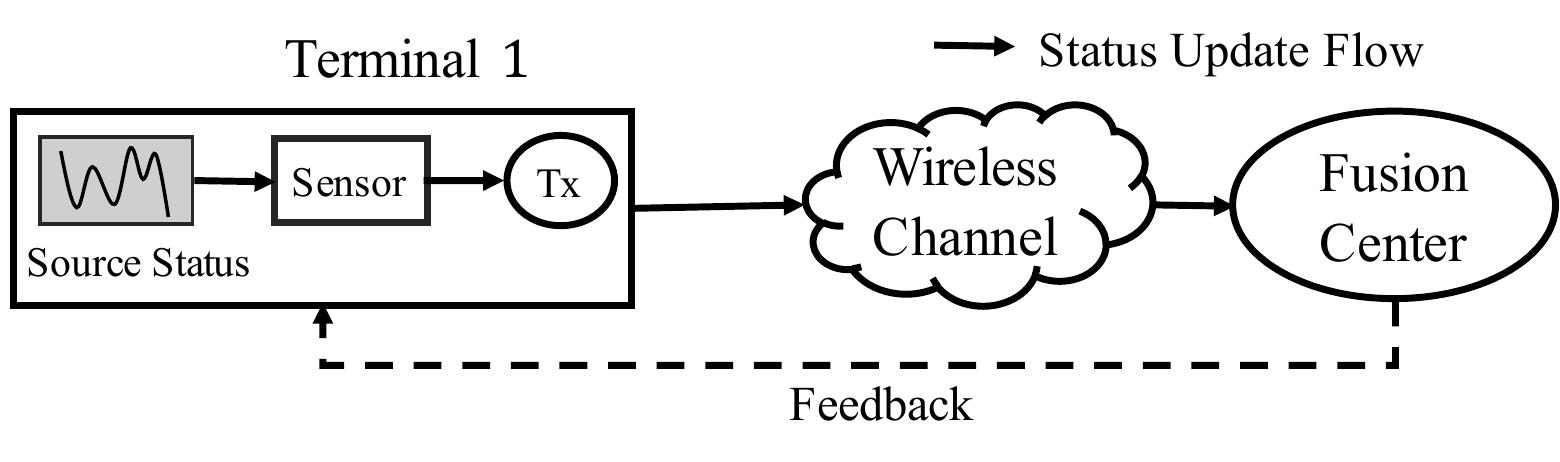}
\caption{A terminal updates its status to a fusion center. }
\label{fig:single}
\end{figure}

Due to the randomness in status evolution, the increment of estimation error at the fusion center is not uniform. It is assumed that the increment $A_t$ of error has zero mean (without loss of generality) and variance $\sigma^2$. It is further assumed that $A_t$ is independent of the current estimation error $Q_t$. An example that satisfies these assumptions is that $A_t$ is an i.i.d Gaussian random variable if the monitored status follows Wiener process. The time-varying context-aware weight $\omega_t$ is assumed to be a random variable that has mean $\bar{\omega}$, and is independent of error $Q_t$. 

To improve the timeliness of status updates, the terminal adaptively allocates its transmission to reduce the average UoI. The average UoI minimization problem is formulated as
\begin{subequations}
\label{program:main-single}
\begin{align}
\min_{U_t} &\quad\limsup_{T\to\infty}\frac{1}{T}\E{\sum_{t=0}^{T-1}\omega_tQ_t^2}\label{obj:single}\\
\mathrm{s.t.}	&\quad \limsup_{T\to\infty}\frac{1}{T}\sum_{t=0}^{T-1}\E{U_t} \leq \rho,\label{eqn:con2}
\end{align}
\end{subequations}
where $Q_{t+1} = \left(1 - U_tS_t\right)Q_t + A_t$. 

\section{Status Updating Scheme}
In order to satisfy the average updating frequency constraint (\ref{eqn:con2}), a virtual queue $H_t$ is defined as $H_{t+1} = \left[H_t - \rho + U_t\right]^+$. When the terminal decides to send a status update, the virtual queue $H_t$ increases by $1-\rho$ as a result; otherwise the virtual queue decreases by $\rho$. Therefore, the length of virtual queue is able to characterize historical usage of wireless and energy resource. More specifically, we have the following lemma:
\begin{lemma}
\label{lemma:mrs}
With $H_0 < \infty$, Eq. (\ref{eqn:con2}) is satisfied as long as virtual queue $H_t$ is mean rate stable, i.e., $\lim_{T\to\infty}\frac{\E{H_T}}{T} = 0. $
\end{lemma}
\begin{proof}
See Appendix B. 
\end{proof}

To obtain a feasible updating scheme, define Lyapunov function as $L_t =  \frac{1}{2}VH^2_t + \theta Q_t^2$, where $\theta$ and $V$ are positive real numbers. The Lyapunov drift is defined as $\Delta_t = \E{L_{t+1} - L_t|Q_t, \omega_{t+1}, H_t}$.

\begin{lemma}
\label{lemma:bound-single}
Denote the penalty at the $t$-th time slot as $f_t$. If $\E{L_0}<\infty$, $\E{f_t}\geq \bar f_\mathrm{min}$ and $\forall t\in\{0,1,2\cdots\}$, $\E{L_{t+1} - L_t + f_t} \leq C$, then the virtual queue $H_t$ is mean rate stable, and
\begin{equation}
\label{eqn:bound-single}
\limsup_{T\to\infty}\frac{1}{T}\sum_{t=0}^{T-1}\E{f_t} \leq C. 
\end{equation}
\end{lemma}
\begin{proof}
See Appendix \ref{sec:lemma2}. 
\end{proof}

Lemma \ref{lemma:bound-single} indicates that the problem of reducing average penalty while keeping the virtual queue stable can be transfered to reducing the sum of expected value of Lynapunov drift and expected penalty at each time slot. As long as the sum is no larger than a constant, the virtual queue is mean rate stable, and the average penalty is upper-bounded. 

\begin{lemma}
\label{lemma:drift-single}
With penalty $f_t = \omega_{t+1}Q_{t+1}^2$, which is the UoI of the next time slot, the sum of Lynapunov drift and expected penalty $f_t$ satisfies
\begin{eqnarray}
\label{eqn:right-single}
&~&\Delta_t + \E{f_t|Q_t, \omega_{t+1}, H_t}\notag\\
&\leq& (\omega_{t+1}+\theta)\sigma^2 + \frac{1}{2}V - V\rho H_t + \omega_{t+1}Q_t^2\notag\\
&~& + (VH_t - (\omega_{t+1}+\theta) pQ_t^2)\E{U_t|Q_t, \omega_{t+1}, H_t}.
\end{eqnarray}
\end{lemma}
\begin{proof}
See Appendix \ref{sec:lemma3}. 
\end{proof}

By selecting $U_t \in \{0,1\}$, we obtain an updating scheme that minimizes the right-hand side of the above inequality at each time slot. 
Next, we optimize parameter $\theta$ to reduce the right-hand side of (\ref{eqn:right-single}). Note that the stationary randomized policy that independently updates with probability $\rho$ at each time slot is also a feasible policy. Substituting the randomized policy into the right-hand side of (\ref{eqn:right-single}) yields
\begin{eqnarray*}
\Delta_t + \E{f_t|Q_t, \omega_{t+1}, H_t}
&\leq& (\omega_{t+1}+\theta)\sigma^2 + \frac{1}{2}V \notag\\
&~&+ \left(\omega_{t+1}(1-p\rho) - \theta p\rho\right)Q_t^2.
\end{eqnarray*}
Since weight $\omega_{t+1}$ is independent to estimation error $Q_t$, taking expectation and letting $\theta = \frac{\bar\omega(1-p\rho)}{p\rho} + \theta_0,\theta_0\geq0$ yields
\begin{eqnarray}
\label{eqn:bound-final-single}
\E{L_{t+1} - L_t + f_t|Q_t}
&\leq& \frac{\bar\omega\sigma^2}{p\rho} + \frac{1}{2}V + \theta_0\sigma^2.
\end{eqnarray}
Since the right-hand side is increasing as $\theta_0$ grows, setting $\theta_0=0$ yields the lowest bound. Then the scheme becomes
\begin{subequations}
\label{p:main-single}
\begin{align}
\min_{U_t}&\quad\left(VH_t - \left(\omega_{t+1}-\bar\omega+\frac{\bar\omega}{p\rho}\right)pQ_t^2\right)U_t\label{eqn:policy-obj}\\
\mathrm{s.t.}&\quad U_t \in \{0,1\}.
\end{align}
\end{subequations}

\begin{proposition}
Define the \emph{update index} $J_t$ as $J_t=\left(\omega_{t+1}-\bar\omega+\frac{\bar\omega}{p\rho}\right)pQ_t^2$.
The solution to policy (\ref{p:main-single}) is
\begin{eqnarray}
\label{eqn:u-single}
U_t=\left\{
\begin{aligned}
&1, &&\mathrm{~if~}J_t > VH_t,\\
&0, &&\mathrm{~if~}J_t \leq VH_t.
\end{aligned}
\right.
\end{eqnarray}
\end{proposition}

The update index is equivalent to the reduction of expected future UoI with a status transmission. According to Eq. (\ref{eqn:u-single}), if the update index is larger than $V$ times virtual queue length $H_t$, there will be a status transmission. In this case, update index is able to measure the necessity of status transmission consider current context, status estimation error, and channel condition, while the virtual queue acts as a dynamic threshold to ensure that the average status update frequency constraint is satisfied. 

\begin{theorem}
\label{thm:single}
Under policy (\ref{eqn:u-single}), the average status update frequency constraint (\ref{eqn:con2}) is satisfied, and the average UoI is upper bounded as
\begin{eqnarray}
\limsup_{T\to\infty}\frac{1}{T}\sum_{t=0}^{T-1}\omega_tQ_t^2\leq\frac{\bar{\omega}\sigma^2}{p\rho} + \frac{1}{2}V. 
\end{eqnarray}
\end{theorem}
\begin{proof}
Take expectation over both sides of (\ref{eqn:bound-final-single}). By the definition of penalty and Lemma \ref{lemma:bound-single}, the theorem is proved. 
\end{proof}

Theorem \ref{thm:single} introduces the UoI upper bound of a status update system with policy (\ref{eqn:u-single}). By Theorem \ref{thm:single}, the smaller parameter $V$ is, the lower the bound is. However, when parameter $V$ is small, policy (\ref{eqn:u-single}) converges slowly, which could lead to performance degradation in a shorter period. 

\section{Numerical Results}
In the simulations, the fusion center receives updates from a terminal whose status evolves as a Wiener process, such that the increments of estimation error are i.i.d. standard Gaussian random variables. 

\begin{figure}[htbp]
\centering
\includegraphics[width=3.4in]{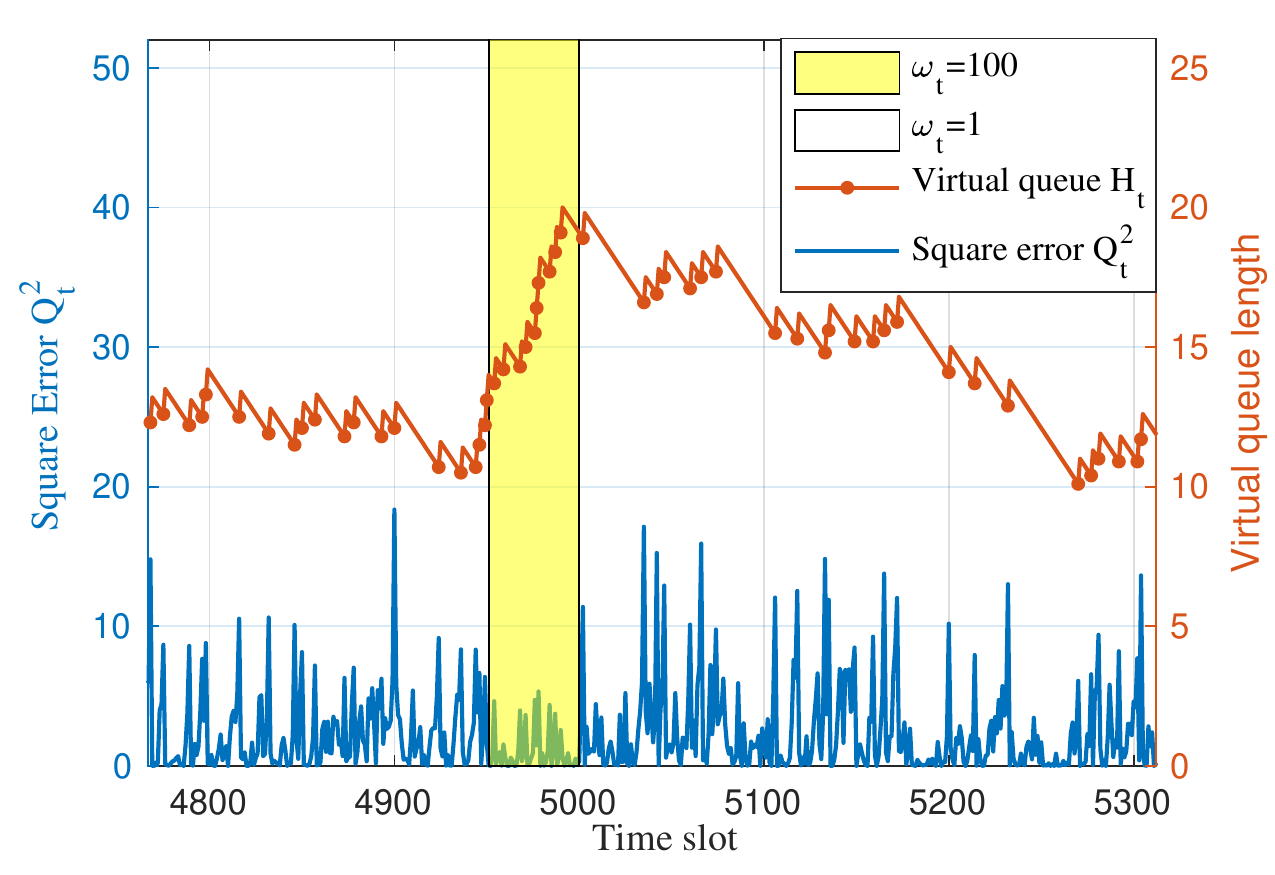}
\caption{Sample path of virtual queue length and squared estimation error. }
\label{fig:ss}
\end{figure}

	In Fig. \ref{fig:ss}, the evolution of virtual queue length and squared estimation error is plotted. The terminal is allowed to transmit in one fourth of the time, i.e., $\rho = 0.25$, with a transmission success probability of 1 for simplicity. Parameter $V$ is set to 1. The context-aware weight is set to $1$ in the former $4950$ time slots, and set to $100$ in the latter $50$ time slots (highlighted in yellow). The difference in context-aware weight indicates that the latter $50$ time slots are the critical period in which the timeliness of status information is faced with a highly strict standard. The solid red dots indicate a status transmission at the corresponding time slots. As shown in the figure, the virtual queue length increases over the critical period, meaning the terminal is updating more frequently than the average frequency to improve the timeliness of status information. Accordingly, the squared estimation error is much lower over the critical period than the ordinary period. 
%
%
%In Fig. \ref{fig:tradeoff}, the average UoI under different status updating frequency constraints and parameter $V$ is plotted. Transmission success probability is set to $0.8$. The context-aware weight at each time slot is i.i.d. with probability $0.01$ being $100$ and probability $0.99$ being $1$. As the figure shows, since a larger $V$ gives a higher priority on guaranteeing the average updating frequency constraint, it can lead to a poorer status update timeliness. Especially when $V$ is larger than $512$ and $\rho$ is larger than $0.5$, the curve is hardly smooth because the actual updating frequency is much lower than the frequency bound $\rho$. In addition, as the available channel resources for status updates increases, the average UoI is reduced, which implies the tradeoff between resource usage and UoI.
%
%\begin{figure}[htbp]
%\centering
%\vspace*{-.1in}
%\includegraphics[width=2.7in]{V-2}
%\vspace*{-.1in}
%\caption{Tradeoff between resource usage and UoI. }
%\label{fig:tradeoff}
%\vspace*{-.1in}
%\end{figure}

In the second set of simulations, transmission success probability is set to $0.8$. The context-aware weight at each time slot is i.i.d. with probability $0.01$ being $100$ and probability $0.99$ being $1$. Since the increment of estimation error is i.i.d. Gaussian random variable, and the context-aware weight $\omega_t$ are i.i.d., by denoting system state as $(Q_t, \omega_t, \omega_{t+1})$, the system is Markovian. Therefore, with relative value iteration \cite{bertsecas}, we are able to obtain the Pareto optimal tradeoff between the average update frequency $\rho$ and average UoI. Similarly, we obtain the AoI-optimal update scheme under each update frequency constraint. Fig. \ref{fig:single-perf} shows the average UoI under AoI-optimal scheme, UoI-optimal scheme, and the proposed adaptive scheme. Although the adaptive scheme has a simple structure and low computation complexity, it achieves a near-optimal performance, while the AoI-optimal scheme yields a much higher UoI. 

\begin{figure}[htbp]
\centering
\includegraphics[width=2.8in]{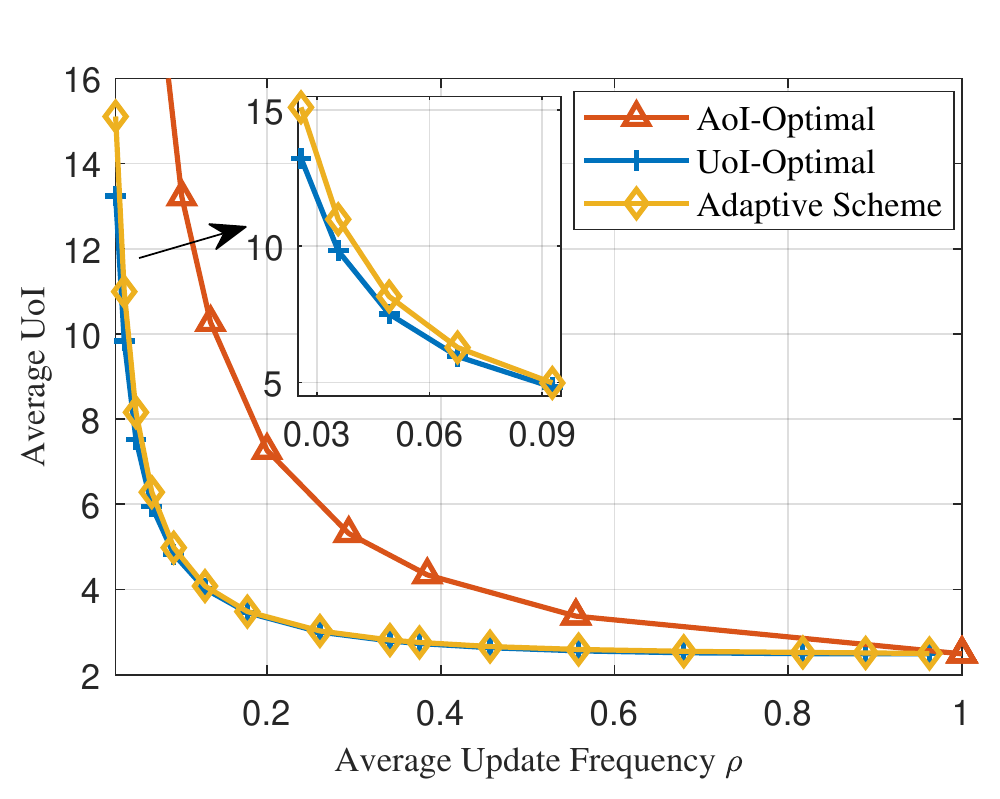}
\caption{UoI under the adaptive, UoI-optimal and AoI-optimal scheme. }
\label{fig:single-perf}
\end{figure}

\section{Conclusions}

This paper introduces the concept of UoI to characterize the timeliness of a status update system. The UoI embeds the context information in status updates as well as the non-uniform status evolution. The status update scheme under an average update frequency constraint is proposed. Simulation results show that the proposed adaptive scheme is able to adapt to the context and status information in the system, and achieve a significant improvement over existing AoI-based schemes on the timeliness of status updates. 

\begin{appendices}
\section{Proof of Proposition 1}
Control actions are made to minimize $$\sum_{t=0}^{T-1}\E{\omega_{t}\left(a\hat{x}_{t-1} + bv_t + r_t - y_t\right)^2}. $$
Taking derivative with respect to $v_t$, the objective function of the problem becomes $2\E{\omega_{t}}\left(a\hat{x}_{t-1} + bv_t - y_t\right)$. Therefore, the minimum weighted squared difference is achieved when 
$v^*_t = \frac{y_t - a\hat{x}_{t-1}}{b}$. 
The objective function of (\ref{program:LQC}) becomes
\begin{eqnarray*}
&~&\limsup_{T\to\infty}\frac{1}{T}\sum_{t=0}^{T-1}\E{\omega_{t}\left(x_t - y_t\right)^2}\\
&=&a^2\limsup_{T\to\infty}\frac{1}{T}\sum_{t=0}^{T-1}\E{\omega_{t}\left(x_{t-1} - \hat{x}_{t-1}\right)^2}\\
&~& + \lim_{T\to\infty}\frac{1}{T}\sum_{t=0}^{T-1}\E{\omega_{t}}\E{r^2}.
\end{eqnarray*}
Since the last term is constant, the proposition is proved.

\section{Proof for Lemma 1}
The definition of virtual queue $H_t$ yields $H_T \geq H_0 - T\rho + \sum_{\tau=0}^{T-1}U_\tau$. Dividing both sides by $T$ and taking limit yields
\begin{eqnarray*}
\lim_{T\to\infty}\frac{\E{H_T} - H_0}{T} &\geq& - \rho + \lim_{T\to\infty}\frac{1}{T}\sum_{\tau=0}^{T-1} \E{U_\tau}.
\end{eqnarray*}
Since the left-hand side equals to zero, the lemma is proved. 

\section{Proof for Lemma 2}
\label{sec:lemma2}
Summing up over $t\in\{0,1,2\cdots,T-1\}$, we get
\begin{eqnarray}
\label{proof:3-1}
\E{L_T} - \E{L_0} + \sum_{t=0}^{T-1}\E{f_t} \leq CT.
\end{eqnarray}
First, we prove that the virtual queue is mean rate stable. By $\E{f_t}\geq \bar f_\mathrm{min}$, we get
$\E{L_T} \leq (C-\bar f_\mathrm{min})T + \E{L_0}$. By the definition of Lynapunov function, we have 
$$\frac{1}{2}V\E{H^2_T} \leq (C-\bar f_\mathrm{min})T + \E{L_0}.$$
Since $\E{H^2_T}\geq\E{H_T}^2$, we obtain 
$$\E{H_T} \leq \sqrt{\frac{2}{V}\left((C-\bar f_\mathrm{min})T + \E{L_0}\right)}.$$
Dividing both sides of the above inequality by $T$ and let $T$ goes to $\infty$, we get $\limsup_{T\to\infty}\frac{\E{H_T}}{T} = 0$.

Next, we prove the upper bound in (\ref{eqn:bound-single}). Dividing (\ref{proof:3-1}) by $T$ and letting $T\to\infty$ yields
\begin{eqnarray*}
\limsup_{T\to\infty}\frac{1}{T}\E{L_T} + \limsup_{T\to\infty}\frac{1}{T}\sum_{t=0}^{T-1}\E{f_t} \leq C.
\end{eqnarray*}
By $L_t\geq0$, the lemma is proved. 

\section{Proof for Lemma 3}
\label{sec:lemma3}
By the dynamic function of $Q_t$ and $U_t,S_t\in\left\{0,1\right\}$, we get
\begin{equation}
Q_{t+1}^2-Q_{t}^2 = A_t^2 + 2\left(1-U_tS_t\right)A_tQ_t - U_tS_tQ_t^2. \label{eqn:drift-q}
\end{equation}
By the definition of virtual queue, we get
\begin{equation}
H_{t+1}^2-H_{t}^2 \leq 1 + 2\left( - \rho + U_t\right)H_t. \label{eqn:drift-h}
\end{equation}
Since $\E{A_t}=0$ and by the assumption that $A_t$ is independent to $Q_t, \omega_{t+1}, H_t$, substituting Eq. (\ref{eqn:drift-q}) and Eq. (\ref{eqn:drift-h}) into the definition of Lyapunov drift yields
\begin{equation}
\label{ineqn:proof1}
\Delta_t \leq \theta\sigma^2 + \frac{1}{2}V - V\rho H_t + (VH_t - \theta pQ_t^2)\E{U_t|Q_t, \omega_{t+1}, H_t}. 
\end{equation}
By adding expected penalty $\E{f_t|Q_t, \omega_{t+1}, H_t}$ to both sides of Eq. (\ref{ineqn:proof1}), the lemma is hereby proved. 

\end{appendices}

\bibliographystyle{ieeetr}

\end{document}